\documentclass[aps,prl,nofootinbib,superscriptaddress,twocolumn]{revtex4-1}

\usepackage[utf8,latin1]{inputenc}
\usepackage[T1]{fontenc}     
\usepackage[american,british]{babel}
\usepackage[dvipsnames]{xcolor}
\usepackage[colorlinks=true,citecolor=Green,linkcolor=Red,urlcolor=Cyan,hyperindex]{hyperref}
\usepackage{graphicx}
\usepackage{epsfig}
\usepackage{color}
\usepackage{centernot}
\usepackage{lmodern}
\usepackage[babel=true]{microtype}
\usepackage{amsmath,amssymb,amsthm}
\usepackage{mathtools}
\usepackage{bm}
\usepackage{enumerate}
\usepackage{sidecap}
\usepackage{cleveref}
\usepackage{enumitem}
\usepackage{listings}

\newcommand{\tr}{\text{\normalfont Tr}}

\newcommand{\ie}{\textit{i.e.}}

\newtheorem{theorem}{Theorem}

\newtheorem{coro}[theorem]{Corollary}
\newtheorem{prop}[theorem]{Proposition}
\theoremstyle{definition}
\newtheorem{example}{Example}

\newcommand{\cH}{\mathcal{H}}

\newcommand{\Herm}{\text{Herm}}
\newcommand{\bD}{\textbf{D}}
\newcommand{\bN}{\textbf{N}}
\newcommand{\bM}{\textbf{M}}

\newcommand{\cD}{\mathcal{D}}
\newcommand{\II}{\mathbb{I}}

\newcommand{\cP}{\mathcal{P}}

\newcommand{\bA}{\textbf{A}}
\newcommand{\cA}{\mathcal{A}}
\newcommand{\cB}{\mathcal{B}}
\newcommand{\cC}{\mathcal{C}}
\newcommand{\bB}{\textbf{B}}
\newcommand{\bS}{\textbf{S}}
\newcommand{\bC}{\textbf{C}}
\newcommand{\cJ}{\mathcal{J}}

\newcommand{\beq}{\begin{equation}}
\newcommand{\eeq}{\end{equation}}

\allowdisplaybreaks

\begin{document}

\title{Uniqueness of the joint measurement and the structure of the set of compatible quantum measurements} 

\author{Leonardo Guerini}
\email[]{leonardo.guerini@icfo.eu}
\affiliation{ICFO-Institut de Ciencies Fotoniques, The Barcelona Institute of Science and Technology, 08860 Castelldefels (Barcelona), Spain}
 \affiliation{Departamento de Matem\'atica, Universidade Federal de Minas Gerais, 31270-901, Belo Horizonte, MG, Brazil}

\author{Marcelo Terra Cunha}
\affiliation{Departamento de Matem\'atica Aplicada, IMECC-Unicamp, 13084-970, Campinas, SP, Brazil}

\date{\today}

\begin{abstract}
We address the problem of characterising the compatible tuples of measurements that admit a unique joint measurement.
We derive a uniqueness criterion based on the method of perturbations and apply it to show that extremal points of the set of compatible tuples admit a unique joint measurement, while all tuples that admit a unique joint measurement lie in the boundary of such a set.
We also provide counter-examples showing that none of these properties are both necessary and sufficient, thus completely describing the relation between joint measurement uniqueness and the structure of the compatible set.
As a by-product of our investigations, we completely characterise the extremal and boundary points of the set of general tuples of measurements and of the subset of compatible tuples.

\end{abstract}


\maketitle

\section{Introduction}\label{intro}

Among the many counter-intuitive features of quantum theory, the fact that we cannot perfectly implement certain measurements concomitantly lies among the most remarkable ones.
This fundamental incompatibility, captured by the notion of joint measurability~\cite{heinosaari16}, lies in the core of a myriad of phenomena and applications of the theory, such as Bell nonlocality~\cite{brunner14}, uncertainty relations~\cite{heisenberg27} and quantum key distribution~\cite{bennett92}.

Similarly to the role played by entanglement in the study of the set of quantum states, measurement compatibility motivates us to investigate the set of quantum measurements.
Although this set is well-understood in terms of extremal and boundary points~\cite{dariano05}, the richer set composed of tuples (ordered sets) of measurements remains to be investigated.
In particular, little is known about the subset of jointly measurable tuples of measurements.

Joint measurability refers to the property of a tuple of measurements to be implemented as a single one, the so-called joint measurement.
Therefore, the joint measurability of such a tuple is equivalent to the \emph{existence} of a joint measurement.
The starting point of this work is to understand the duality existence-uniqueness in this case, posing the follow up question: what compatible tuples of measurements admit a \emph{unique} joint measurement?

In the particular case of compatible tuples of projective measurements we do have a unique joint measurement.
Joint measurement uniqueness can be further connected to the concepts of greatest and maximal lower bounds~\cite{teiko08}.
The extremality of the measurements in the tuple was also studied, and found to be related to the extremality and uniqueness of the corresponding joint measurement~\cite{haapasalo14}, but not equivalent.
This property is also sufficient for some relations between joint measurability and coexistence~\cite{haapasalo15}, but the extremality of the tuple itself was never considered.
Since the set of jointly measurable tuples is a convex proper subset of the set of general tuples of measurements, a natural idea is to investigate the relation between joint measurement uniqueness and the extremality/boundary property of tuples in this set (Figure \ref{scyther}).

In this work we characterise the extremal and boundary points of the sets of general and compatible tuples of measurements.
We prove that a compatible tuple is extremal (respectively, in the boundary) in the compatible set if and only if its joint measurement is extremal (respectively, in the boundary) in the corresponding set of single measurements.

Furthermore, we generalise the perturbation technique introduced in Ref.~\cite{dariano05} and derive a criterion for deciding whether a given compatible tuple admits a unique joint measurement.
We apply it to prove necessary and sufficient conditions for joint measurement uniqueness, namely that (i) extremal tuples (in the compatible set) admit a unique joint measurement, and (ii) tuples admitting a unique joint measurement lie in the boundary of the compatible set.
We also present examples showing that none of this conditions is both necessary and sufficient, therefore completely determining the relation between uniqueness and the structure of the compatible set.

\begin{figure}[h!]
\begin{center}
\includegraphics[width=0.9\linewidth]{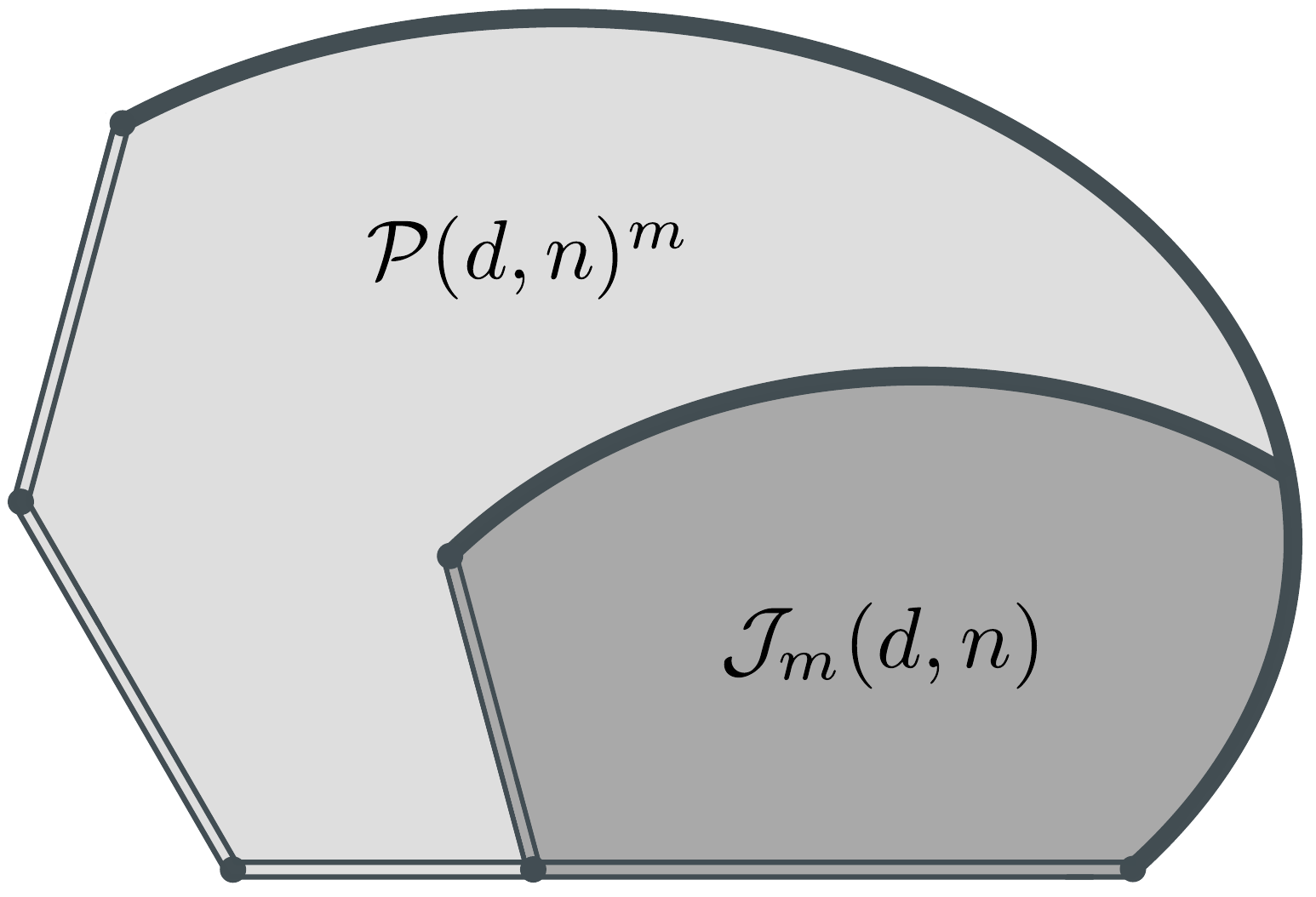}
	\caption{\small{\label{scyther}The two sets of main interest in this work: the set $\cP(d,n)^{ m}$ of tuples of $m$ measurements in dimension $d$ with $n$ outcomes (in light colour) and its subset $\cJ_m(d,n)$ of jointly measurable tuples (in dark colour).
The boundaries of both sets are composed of extremal (in full line) and non-extremal points (in double line).}}
\end{center}
\end{figure}

\section{Preliminaries}

Let $\cH$ be a $d$-dimensional Hilbert space and $\Herm(\cH)$ be the set of Hermitian operators acting on $\cH$.
A \textit{quantum measurement} on $\cH$ with $n$ outcomes is modelled by a POVM (positive-operator-valued measure), which is a tuple $\bA = (A_1,  \ldots , A_n)\in \Herm(\cH)^n$ of positive semi-definite operators satisfying $\sum_i{A_i} = \II$, where each $A_i$ corresponds to outcome $i$ and $\II$ is the identity operator on $\cH$.
The operators $A_i$ are called the \textit{effects} of $\bA$.
In the case where the effects $A_i$ are projectors, we say that $\bA$ is a \textit{projective measurement}.
Notice that some effects might be null, corresponding to outcomes that never occur.

We denote the set of $n$-outcome qudit measurements by $\cP(d,n)$, or simply by $\cP$ if the dimension and number of outcomes are redundant.
$\cP$ is a convex set.
Indeed, given $\bA, \bB\in\cP$ and $p\in[0,1]$, their convex combination $p\bA+(1-p)\bB=:\bC\in\cP$ is the measurement defined by the effects $C_i=pA_i+(1-p)B_i$, for $i=1,\ldots,n$.
A measurement is \textit{extremal} (in $\cP$) if it cannot be decomposed into the convex combination of two different POVMs.

A tuple~\footnote{Here we consider tuples of measurements (instead of sets) in order to unambiguously associate the POVM on the $j$-th entry of the tuple with the marginal obtained by summing the effects of the joint measurement over all indices but the $j$-th. Tuples are also the appropriate object to define extremality of jointly measurable collections of measurements.} of $m$ $n$-outcome measurements $\cA = [\bA^{(1)},\ldots,\bA^{(m)}]\in\cP(d,n)^{ m}$ is \textit{jointly measurable}, or \textit{compatible}, if there exists a \textit{joint measurement} $\bM = (M_{a_1 \ldots a_m})$, with $a_i\in\{1, \ldots ,n\}$ for all $i$, such that
\begin{subequations}\label{bulbasaur}
\begin{align}
&M_{a_1 \ldots a_m} \geq 0, \ \forall a_1, \ldots ,a_m \\
\sum_{s\neq j}\sum_{a_s=1}^n &M_{a_1 \ldots (a_j=i) \ldots a_m} = A^{(j)}_i, \ \forall j, i.
\end{align}
\end{subequations}
Hence, all POVM elements $A^{(j)}_{i}$ can be recovered by coarse-graining over $\bM$.
Notice that the normalisation $\sum_s\sum_{a_s}M_{a_1, \ldots ,a_m}=\II$ is guaranteed by the normalisation of any of the POVMs once the marginal constraints are satisfied, thus $\bM$ is a valid POVM.
We denote by $\cJ_m(d,n)$ the subset of $\cP(d,n)^{m}$ composed by jointly measurable tuples of $m$ POVMs.
Again, for ease of notation, we will use simply $\cP^m$ and $\cJ_m$ whenever the dimension and number of outcomes are redundant.

One can write a feasibility SDP (semidefinite programme) to decide whether a given tuple of POVMs $\cA$ is jointly measurable~\cite{wolf09}:
\begin{align}\label{squirtle}
\nonumber \mathrm{given}\ \ &\cA = [\bA^{(j)}] \\ \mathrm{find}\ \ &\bM = (M_{a_1 \ldots a_m}) \\
\nonumber \mathrm{s.t.}\ \ &M_{a_1 \ldots a_m} \geq 0,\ \forall a_1, \ldots ,a_m\\
\nonumber &\sum_{s\neq j}\sum_{a_s=1}^n M_{a_1 \ldots (a_j=i) \ldots a_m} = A^{(j)}_i,\ \forall j, i.
\end{align}
An SDP formulation of a problem is valuable since this is a class of problems that can be solved computationally in an efficient way.

We can depolarise a POVM $\bA$ by applying the \textit{depolarising map} 
\beq
\Phi_t: A_i \mapsto tA_i + (1-t)\frac{\tr(A_i)}{d}\II
\eeq
to each effect $A_i$, where $t\in [0,1]$ is called the \emph{visibility} of the depolarised POVM.
Hence we write
\beq
\Phi_t(\bA) := (\Phi_t(A_1),\ldots,\Phi_t(A_n)).
\eeq
Notice that depolarising $\bA$ is equivalent to mixing it with a trivial POVM having all effects proportional to the identity,
\beq\label{pidgey}
\Phi_t(\bA) = t\bA + (1-t)\bA^{\tr},
\eeq
where $\bA^{\tr}=(\tr(A_1)\II/d, \ldots ,\tr(A_n)\II/d)$.

By depolarising each POVM in a tuple $\cA=[\bA^{(j)}]$ of measurements we obtain a depolarised tuple
\beq
\Phi_t(\cA) := [\Phi_t(\bA^{(j)})].
\eeq 
Every tuple $\cA$ becomes jointly measurable if depolarised enough, \ie, for sufficiently small $t$.
It is straightforward to modify SDP (\ref{squirtle}) to find the maximum depolarisation parameter $t$ that makes a target tuple $[\bA^{(j)}]$ jointly measurable \cite{cavalcanti17}.

We are interested in characterising the tuples that admit a unique joint measurement.
Notice that if $\bM, \bM'$ are two joint measurements for a fixed tuple of measurements $\cA$, then any convex combination $\bN=p\bM+(1-p)\bM'$ satisfies Eqs. (\ref{bulbasaur}), showing that $\bN$ is also a joint measurement for this tuple.
Therefore the set of joint measurements for $\cA$ is convex.
In particular, since every parameter $p\in[0,1]$ provides a different joint $\bN$, we have that for every tuple of POVMs there are exactly zero (in the incompatible case), one or infinitely many joint measurements for it.

\section{Perturbations for POVMs}

We approach the problem of deciding whether a given compatible tuple of POVMs admits a unique joint measurement by adapting the perturbation method introduced in Ref. \cite{dariano05}.
For every measurement $\bA\in\cP(d,n)$, there are infinitely many tuples $\bD^\bA\in\Herm(\cH)^n$ that preserves its POVMness, \ie, satisfying
\begin{subequations}\label{bellsprout}
\begin{align}
A_{i} + D^A_{i} &\geq 0,\quad i=1,\ldots,n \\ \label{weepinbell}
\sum_i{D_i^A} &= 0
\end{align}
\end{subequations}
and therefore ensuring that $\bA + \bD^\bA$ is still a valid POVM.
Notice that we require no positive semi-definitiveness from the perturbation operators (nor could they all be positive semi-definite without being trivial, due to Eq. (\ref{weepinbell})).
We say that such tuples $\bD^\bA$ are \textit{POVM-preserving perturbations for $\bA$}.

\begin{example}
We can interpret a depolarised POVM $\Phi_t(\bA)$ as a perturbed version of $\bA$, as already hinted in Eq. (\ref{pidgey}).
For every $\bA\in\cP(d,n)$, the tuple $\bD^\tr$ defined by 
\beq\label{magikarp}
D_i^\tr =\frac{\tr(A_i)}{d}\II - A_i,
\eeq
for $i=1,\ldots,n$, satisfies Eqs. (\ref{bellsprout}), and also
\beq
\Phi_t(\bA) = \bA + (1-t)\bD^\tr.
\eeq
Therefore, $(1-t)\bD^{\tr}$ is a POVM-preserving perturbation for $\bA$ for any $t\in [0,1]$.
We call $\bD^\tr$ the \textit{depolarising perturbation} for $\bA$.
\end{example}

As already pointed in Ref. \cite{dariano05}, we can apply the perturbation method to decide on the extremality of a measurement.
Namely, a POVM $\bA$ is extremal if and only if it does not possess any non-null perturbation $\bD^\bA$ that is symmetric, in the sense that $-\bD^\bA$ is also a POVM-preserving perturbation for $\bA$.
Otherwise, there would be the convex decomposition
\beq
\bA = \frac12\left[(\bA+\bD^\bA) + (\bA-\bD^\bA)\right],
\eeq
witnessing the non-extremality of $\bA$.

The set $\cP(d,n)^{ m}$ of tuples of measurements inherits the convex structure of $\cP(d,n)$ once we define the convex combination componentwise.
Hence, we can introduce perturbations for a tuple of POVMs.
We say that $\cD^\cA=[\bD^{(j)}]_j\in\Herm(\cH)^{nm}$ is a \textit{POVM-preserving perturbation for $\cA$} if each $\bD^{(j)}$ is a POVM-preserving perturbation for $\bA^{(j)}$.
Thus, $\cA$ is an extremal tuple of POVMs (in $\cP(d,n)^{ m}$) if and only if there is no symmetric perturbation $\cD^\cA\neq 0$ for $\cA$.
Consequently, the extremality of $\cA$ is strongly bonded to the extremality of its elements.
For sake of completeness we prove the following characterisation, that holds for any Cartesian product such as $\cP(d,n)^{m}$.

\begin{theorem}\label{kubone}
A tuple of measurements $\cA=[\bA^{(1)},\ldots,\bA^{(m)}]\in\cP^{ m}$ is extremal in $\cP^{ m}$ if and only if each of its elements $\bA^{(j)}, j=1,\ldots,m,$ is extremal in $\cP$.
\end{theorem}
\begin{proof}
We prove both directions by contraposition.
Suppose, without loss of generality, that $\bA^{(1)}$ is not extremal, and that $\bD\neq 0$ is a symmetric perturbation for it.
Hence, $\cD=[\bD,0,\ldots, 0]$ is a non-null symmetric perturbation for $\cA$, which therefore is not extremal.

Conversely, if $\cA$ is not extremal in $\cP^{ m}$, then there exists a non-null symmetric perturbation $\cD=[\bD^{(j)}]$.
Since $\bD^{(j_0)}$ is non-null for some $j_0$ and by definition it is a symmetric perturbation for $\bA^{(j_0)}$, we conclude that $\bA^{(j_0)}$ is not extremal.
\end{proof}

The set of main interest in this work is $\cJ_m$, the subset of $\cP^{ m}$ of jointly measurable tuples of measurements.
Given $\cA, \cB \in \cJ_m$, consider their joint measurements $\bM^\cA$ and $\bM^\cB$, respectively.
Any convex combination $\cC = t\cA + (1-t)\cB$ admits a joint measurement given by $\bM^\cC = t\bM^\cA + (1-t)\bM^\cB$, and therefore $\cJ_m(d,n)$ is a convex set.
Following the perturbation approach, we say that $\cA$ is extremal in $\cJ_m$ if there is no \textit{joint measurability-preserving (JM-preserving) perturbation} $\cD^\cA$ such that $\cA\pm\cD^\cA\in\cJ_m$ (see Figure \ref{spearow}).

\begin{figure}[h!]
\begin{center}
\includegraphics[width=0.9\linewidth]{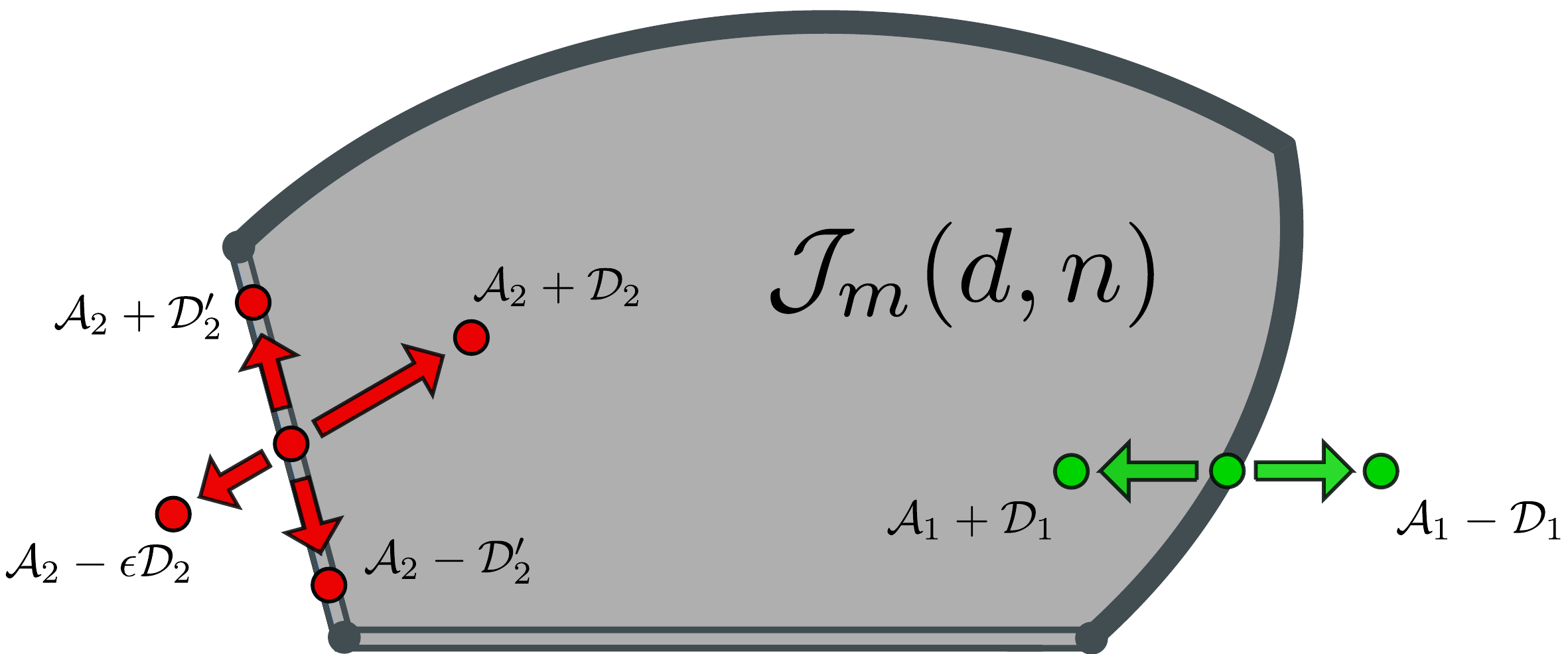}
	\caption{\small{\label{spearow}In green, an extremal compatible tuple $\cA_1$, for which any POVM-preserving perturbation ${\cD_1}$ satisfies $\cA_1-{\cD_1}\notin\cJ_m(d,n)$.
	In red, the compatible tuple $\cA_2$ lies in the boundary $\partial(\cJ_m(d,n))$, as witnessed by the perturbation $\cD_2$ (see Section \ref{jotoleague}), but it is not extremal, as witnessed by the perturbation $\cD_2'$.
}}
\end{center}
\end{figure}

Analogously to a perturbation that preserves POVMness, we can study perturbations that preserve the marginals of a joint measurement.
More specifically, given a joint measurement $\bM$ for a tuple $\cA$, we can search for a \textit{marginal-preserving perturbation} $\bD^\bM$ such that $\bM + \bD^\bM$ is still a joint measurement for $\cA$.
This is equivalent to say that
\begin{subequations}\label{geodude}
\begin{eqnarray}
M_{a_1 \ldots a_m} + D^M_{a_1 \ldots a_m} &\geq& 0, \ \forall a_1, \ldots ,a_m,\\ 
\sum_{s\neq j}\sum_{a_s=1}^n{D^M_{a_1 \ldots (a_j=i) \ldots a_m}} &=& 0, \ \forall j, i.
\end{eqnarray}
\end{subequations}
Notice that Eqs. (\ref{geodude}) also ensure that $\bD^\bM$ preserves the POVMness of $\bM$, hence a marginal-preserving perturbation is a particular case of a POVM-preserving perturbation.

As before, one can write an SDP to decide whether there is a perturbation that preserves the marginals for a given joint measurement $\bM$,
\begin{align}
\nonumber 
\mathrm{given} \ \ & \bM \\ \label{charmander}
\max_\bD \ \ & \text{tr}(D_{1 \ldots 1}\lambda_1)\\ \nonumber
\mathrm{s.t.} \ \ & M_{a_1 \ldots a_m} + D_{a_1 \ldots a_m} \geq 0, \ \forall a_1, \ldots ,a_m,\\
\nonumber
& \sum_{s\neq j}\sum_{a_s=1}^n{D_{a_1 \ldots (a_j=i) \ldots a_m}} = 0, \ \forall j, i,
\end{align}
where $\{\lambda_i\}$ is an orthogonal basis of the space of Hermitian operators and $D_{1\ldots1}$ is the first perturbation operator.
We write it as a maximisation of the first coefficient in the decomposition $D_{1 \ldots 1}= \sum_{i=1}^{d^2}\alpha_i\lambda_i$ only for ensuring a non-null optimal argument $\bD$ (in the case where there is one).
Hence the complete computational test for checking on the uniqueness of a joint measurement comprehends running the above SDP for each basis operator $\lambda_i$ and each perturbation operator $D_{d_1 \ldots d_n}$ in the objective function.

The above formulation provides a criterion for deciding on the uniqueness of the joint measurement $\bM$ for a given tuple $\cA$.

\begin{theorem} \label{snorlax}
Let $\cA$ be a jointly measurable tuple of measurements and $\bM$ a joint measurement for $\cA$.
Then $\bM$ is unique if and only if there is no marginal-preserving perturbation $\bD\neq 0$ for $\bM$.
\end{theorem}

\begin{proof}
If there is such non-null perturbation $\bD$, then $\bM'= \bM + \bD$ is also a joint measurement for $\cA$.
On the other hand, if there is another joint measurement $\bM'\neq \bM$, then define the non-null perturbation $\bD = \bM'-\bM$.
\end{proof}

\section{Extremality and uniqueness}

Now that we have a criterion to decide on the uniqueness of joint measurements, in this section we study its relation with the extremality of the tuple of POVMs $\cA$ in question.
More specifically, we prove that the extremality of the tuple in the compatible set is sufficient to guarantee the uniqueness of the joint measurement.
The proof is based on the idea of constructing a JM-preserving perturbation for a compatible tuple starting from a marginal-preserving perturbation for its joint measurement.

\begin{theorem}\label{gyarados} Let $\cA \in\cJ_m$ be a jointly measurable tuple of measurements.
If $\cA$ is extremal in $\cJ_m$, then there exists a unique joint measurement for $\cA$.
\end{theorem}

\begin{proof}
We prove the contrapositive of the above statement, namely that if the joint measurement is not unique, then $\cA$ is not extremal in $\cJ_m$.
More specifically, we provide an algorithmic method that finds within $m$ steps a non-trivial decomposition of $\cA$ into compatible tuples whenever the joint measurement is not unique.

Suppose that $\bM,\bM'$ are two distinct joint measurements for $\cA$, and define $\bD = \bM'- \bM$.
Hence the perturbation $\bD$ satisfies
\begin{subequations}
\begin{align}
M_{a_1 \ldots a_m} + D_{a_1 \ldots a_m} \geq 0 , \ \forall a_1, \ldots ,a_n\\
M'_{a_1 \ldots a_m} - D_{a_1 \ldots a_m} \geq 0 , \ \forall a_1, \ldots ,a_n\\ \label{gastly}
\sum_{a_r\neq a_j}{D_{a_1 \ldots (a_j=i) \ldots a_m}} = 0,  \ \forall j, i.
\end{align}
\end{subequations}

Since $\bM\neq \bM'$, there is at least one non-null perturbation operator $D_{a_1 \ldots a_m}$.
Without loss of generality, lets assume $D_{1 \ldots 1}\neq 0$; otherwise we could consider a relabeling of the POVMs in $\cA$ whose joint measurements satisfy this property, and both joint measurability and number of joint measurements are preserved by relabelings.
Since each marginal of $\bD$ vanishes, our strategy is to use them as POVM-preserving perturbations for $\bM$ and $\bM'$ in order to construct a symmetric JM-preserving perturbation for $\cA$.

Consider the POVMs $\bM^{(1,+)}, \bM^{(1,-)}$ given by effects 
\begin{subequations}\label{kakuna}
\begin{align}
M^{(1,+)}_{a_1 \ldots a_m} = M_{a_1 \ldots a_m} + \delta_{1,a_1}D_{a_1 \ldots a_m} \\
M^{(1,-)}_{a_1 \ldots a_n} = M'_{a_1 \ldots a_m} - \delta_{1,a_1}D_{a_1 \ldots a_m}
\end{align}
\end{subequations}
where $\delta_{1,a_1}$ equals 1 if $a_1=1$ and 0 otherwise.
These operators are positive semi-definite and $\bM^{(1,+)}, \bM^{(1,-)}$ are normalised due to Eq. (\ref{gastly}).
Hence, their marginals define valid POVMs, given by
\begin{subequations}\begin{align}\nonumber
\sum_{s\neq j}&\sum_{a_s=1}^n {M^{(1,\pm)}_{a_1\ldots (a_j=i)\ldots a_m}} \\
&= A^{(j)}_i \pm (1-\delta_{1,j})\sum_{s\neq 1,j}\sum_{a_s=1}^n D_{(a_1=1)\ldots(a_j=i)\ldots a_m} \\
\label{ponyta}
&=: A^{(j,1,\pm)}_i.
\end{align}\end{subequations}
These define valid tuples of POVMs 
\beq
\cA^{(1,\pm)} = [\bA^{(j,1,\pm)}]_j
\eeq
which are jointly measurable by definition, since each was born from the joint measurement $\bM^{(1,\pm)}$.
Thus we obtained the decomposition
\beq
\cA = \frac12 (\cA^{(1,+)} + \cA^{(1,-)}).
\eeq
If $\cA^{(1,\pm)}\neq \cA$, this decomposition is non-trivial and we are done.
If not, we have that for all $j=1,\ldots,m$, $\bA^{(j,1,\pm)} = \bA^{(j)}$ and in particular $A^{(2,1,\pm)}_1 = A^{(2)}_1$, which implies
\begin{equation}
\sum_{s>2}\sum_{a_s=1}^n D_{11a_3\ldots a_m} = 0.
\end{equation}
Taking this marginal as perturbation, we construct now the POVMs
\begin{subequations}
\begin{align}
M^{(2,+)}_{a_1 \ldots a_m} = M_{a_1 \ldots a_m} + \delta_{1,a_1}\delta_{1,a_2}D_{a_1 \ldots a_m} \\
M^{(2,-)}_{a_1 \ldots a_n} = M'_{a_1 \ldots a_m} - \delta_{1,a_1}\delta_{1,a_2}D_{a_1 \ldots a_m},
\end{align}
\end{subequations}
whose marginals define effects 
\begin{align}\nonumber
A^{(j,2,\pm)}_i :=& A^{(j)}_i \pm (1-\delta_{j,1})(1-\delta_{j,2})\\
&\times \sum_{s\neq 1,2,j}\sum_{a_s=1}^n D_{11a_3\ldots(a_j=i)\ldots a_m}
\end{align}
and tuples $\cA^{(2,\pm)}=[\bA^{(j,2,\pm)}]_j$, satisfying $\cA = (\cA^{(2,+)}+\cA^{(2,-)})/2$.
Again, if $\cA^{(2,\pm)}\neq \cA$, we are done. 
If not, we have $A^{(3,2,\pm)}_1 = A^{(3)}_1$ and
\begin{equation}
\sum_{s>3}\sum_{a_s=1}^n D_{111a_4\ldots a_m} = 0,
\end{equation}
for all $j=3,\ldots,m,\ i=1,\ldots,n$.
We use now this marginal as perturbation for $\bM, \bM'$ and construct $\bM^{(3,\pm)}$, repeating the process.
Thus at each round $k$ we can construct a decomposition into compatible tuples $\cA = (\cA^{(k,+)}+\cA^{(k,-)})/2$, that if not trivial concludes the proof, and otherwise provides a new restriction
\begin{equation}
\sum_{s>k+1}\sum_{a_s=1}^n D_{1\ldots1a_{k+2}\ldots a_{m}} = 0,
\end{equation}
that is used as perturbation for $\bM,\bM'$ in the next round.

If after $m-2$ repetitions all yielded decompositions are trivial, we take the last step and obtain POVMs given by
\begin{equation}
A^{(j,m-1,\pm)}_i = A^{(j)}_i \pm \delta_{j,m}D_{1\ldots1i},
\end{equation}
where we used $\delta_{j,m} = (1-\delta_{j,1})\ldots(1-\delta_{j,m-1})$, which compose compatible tuples $\cA^{(m-1,\pm)}=[\bA^{(j,m-1,\pm)}]_j$.
These form a decomposition for $\cA$ for which we guarantee that $\cA^{(m-1,\pm)}\neq \cA$, since 
\begin{equation}
D_{1\ldots1}\neq0 \implies A^{(m,m-1,\pm)}_1 \neq A^{(m)}_1.
\end{equation}
Hence, $\cA$ is not extremal in $\cJ_m$.
\end{proof}

The proof of Theorem \ref{gyarados} allows for the following corollaries, further relating properties of the tuple with features of the joint measurement.
Corollary \ref{exeggcutor} and a weaker version of Corollary \ref{exeggcute} (restricted to $m=2$ and requiring the tuple to be extremal both on $\cP^{m}$ and $\cJ_m$) were already proved in Ref. \cite{haapasalo14}, in the context of operator algebras.

\begin{coro}\label{exeggcutor}
Let $\cA = [\bA^{(1)},\ldots,\bA^{(m)}]\in\cJ_m$ be a jointly measurable tuple of measurements.
If any $\bA^{(j)}$ in $\cA$ is an extremal measurement, then $\cA$ admits a unique joint measurement.
\end{coro}

\begin{proof}
The proof of Theorem \ref{gyarados} shows that if there is more than one joint measurement for the tuple, there is an algorithm that provides a non-trivial decomposition $\bA^{(m)} = (\bA^{(m,j,+)} + \bA^{(m,j,-)})/2$, where in the worst case $j=m-1$.
The latter is a direct consequence of the fact that at each step $k$ of the algorithm we choose $A^{(k+1)}_1=A^{(k+1,k,\pm)}_1$ to yield the next perturbation and $k=m-1$ is the last choice, but any effect $A^{(l)}_1,\ l=k+1,\ldots,m,$ provides an analogous condition.
If we alter the recipe such that the condition corresponding to $A^{(j)}_1$ is chosen at last, for any $j=2,\ldots,m$, it culminates in a non-trivial decomposition for $\bA^{(j)}$.
We extend this to $j=1$ by placing $\delta_{1,a_2}$ in (\ref{kakuna}), to kickstart the process.
Therefore, various joint measurements for $\cA$ allow us to obtain a non-trivial decomposition for any measurement $\bA^{(j)}$ in the tuple.
\end{proof}

\begin{coro}\label{exeggcute}
$\cA = [\bA^{(1)},\ldots,\bA^{(m)}]\in\cJ_m(d,n)$ is extremal in $\cJ_m(d,n)$ if and only if its (unique) joint measurement is extremal in $\cP(d,n^m)$.
\end{coro}

\begin{proof}
We start by proving the "only if" part.
Suppose $\bM$ is a joint measurement for $\cA$ and that exists a perturbation $\bD$ for $\bM$ such that $\bM\pm \bD \in \cP(d,n^m)$.
If $\bD$ is a marginal-preserving perturbation, then $\cA$ admits more than one joint measurement and by Theorem \ref{gyarados} we have that $\cA$ is not extremal.
It is left to prove now the case in which $\bD$ is not marginal-preserving.

For each $j\in\{1,\ldots,m\}$ and $i\in\{1,\ldots,n\}$, consider 
\beq
A^{(j,\pm)}_{a_j}=\sum_{r\neq j}\sum_{a_r}M_{a_1\ldots (a_j=i)\ldots a_m} \pm D_{a_1\ldots (a_j=i)\ldots a_m}.
\eeq
This defines measurements $\bA^{(j,\pm)}$, and since $\bD$ is not marginal-preserving we have 
\beq
\sum_{r\neq j}\sum_{a_r}D_{a_1\ldots (a_j=i)\ldots a_m}\neq 0
\eeq
for at least one choice of $j$ and $i$.
Hence, $\bA^{(j,\pm)}\neq \bA^{(j)}$ and for $\cA^{(\pm)} = [\bA^{(j,\pm)}]_j$ we have the non-trivial decomposition
\beq
\cA = \frac 12 (\cA^{(+)} + \cA^{(-)}).
\eeq 
Since $\cA^{(+)},\cA^{(-)}\in\cJ_m(d,n)$, we conclude that $\cA$ is not extremal in $\cJ_m(d,n)$.

Conversely, suppose that $\cA$ is not extremal in $\cJ_m(d,n)$.
Then there exist compatible tuples $\cA^{(+)}\neq \cA^{(-)}\in\cJ_m(d,n)$ such that $\cA=(\cA^{(+)} + \cA^{(-)})/2$.
Denoting by $\bM^{(+)}$ and $\bM^{(-)}$ their joint measurements, respectively, we see that $\bM^{(+)}\neq\bM^{(-)}$ and the measurement defined by $\bM:=(\bM^{(+)} + \bM^{(-)})/2$ is a joint measurement for $\cA$ that is clearly not extremal in $\cP(d,n^m)$.
\end{proof}

\begin{figure}[h!]
\begin{center}
\includegraphics[width=1\linewidth]{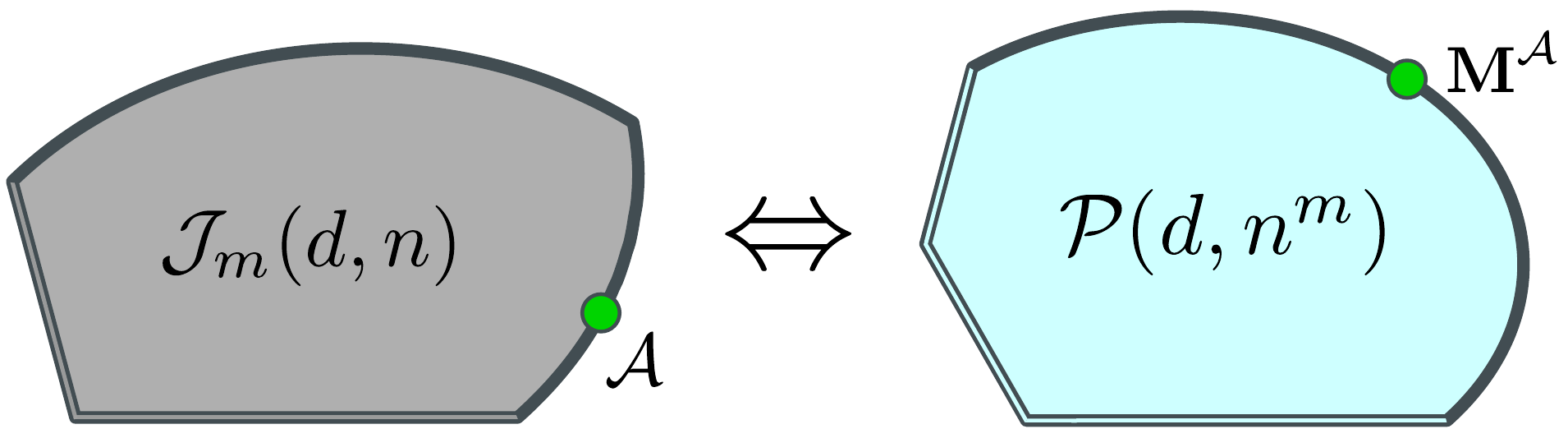}
	\caption{\small{A pictorial representation of Corollary \ref{exeggcute}: a tuple of measurements $\cA$ is extremal in the compatible set $\cJ_m(d,n)$ if and only if its unique joint measurement $\bM^\cA$ is extremal in the set $\cP(d,n^m)$.}}
\end{center}
\end{figure}

Since extremal measurements in dimension $d$ have at most $d^2$ non-null effects~\cite{dariano05}, Corollary \ref{exeggcute} implies that any joint measurement with more than $d^2$ non-null effects corresponds to non-extremal tuples.

The following example illustrates Corollary \ref{exeggcutor} and at the same time shows that the converse of Theorem \ref{gyarados} is false.

\begin{example}\label{lickitung}
Consider the pair of trivial POVMs $\cA = [(\II, 0), (\II/2, \II/2)]\in\cJ_2(2,2)$, which can be decomposed in two trivially compatible tuples,
\beq
\cA = \frac12\left([(\II,0),(\II,0)]+[(\II,0),(0,\II)]\right),
\eeq
and therefore is not extremal in $\cJ_2(2,2)$.
A joint measurement $\bM$ for $\cA$ should satisfy $M_{21}+M_{22} = 0$, hence $M_{21}=0=M_{22}$.
Since the other marginals imply $M_{11} = M_{11} + M_{21}=\II/2$, and analogously $M_{12}= M_{12} + M_{22} = \II/2$, we see that these conditions define uniquely each effect of the joint measurement $\bM$.
Therefore, not all tuples that admit a unique joint measurement are extremal.
\end{example}

However, every time we have a tuple of trivial POVMs and at least two of them are not deterministic, we can show that it admits an infinite number of joint measurements.
\begin{prop}\label{electabuzz}
Let $\cA=[\vec p^{(1)}\II,\ldots,\vec p^{(m)}\II]$ be a tuple of trivial POVMs, where each $\vec p^{(j)}$ is an $n$-dimensional probability vector and $\vec p^{(j)}\II := (p^{(j)}_1\II, \ldots ,p^{(j)}_n\II)$.
If at least two vectors $\vec p^{(j_0)}, \vec p^{(j_1)}$ are not deterministic, \ie, $\vec p^{(j_0)}, \vec p^{(j_1)}\in [0,1)^{ n}$, then $\cA$ has an infinite number of joint measurements.
\end{prop}

\begin{proof}
Note that every tuple of trivial POVMs $\cA$ admits a joint measurement $\bM$ defined by
\beq\label{onix}
M_{a_1 \ldots a_m} = \prod_{s=1}^m{p_{a_s}^{(s)}\II}.
\eeq 
Without loss of generality, assume that $\vec p^{(1)}, \vec p^{(2)}$ are non-deterministic, and hence each vector has two non-null entries, say $p^{(1)}_1, p^{(1)}_2, p^{(2)}_1, p^{(2)}_2 > 0$.
Furthermore, assume that $p^{(j)}_1 >0,\ \forall j>2$ (since some entry of every $\vec p ^{(j)}$ must be non-null).

Let us now construct a marginal-preserving perturbation for $\bM$.
Consider $\alpha =$ $\min \{p^{(1)}_a p^{(2)}_b p^{(3)}_1  \ldots  p^{(m)}_1 ; a,b\in\{1,2\}\}$.
Given our assumptions, this set contains only strictly positive elements, and therefore $\alpha>0$.
Notice now that
\begin{subequations}
\begin{align}
M_{111 \ldots 1} + \alpha\II &\geq 0 \\
M_{121 \ldots 1} - \alpha\II &\geq 0 \\
M_{211 \ldots 1} - \alpha\II &\geq 0 \\
M_{221 \ldots 1} + \alpha\II &\geq 0.
\end{align}
\end{subequations}
This shows that the perturbation $\bD$ defined by
\beq\nonumber
D_{a_1 \ldots a_m} =
\begin{cases}
\alpha \II, &\text{if}\ \ a_1 \ldots a_m \in \{111 \ldots 1, 221 \ldots 1\} \\
-\alpha \II, &\text{if}\ \ a_1 \ldots a_m \in \{121 \ldots 1, 211 \ldots 1\} \\
0, &\text{otherwise}\ \ 
\end{cases}
\eeq
preserves the marginals of $\bM$, since it is straightforward to check that the marginals of $\bD$ sum up to zero.
We now apply Proposition \ref{snorlax} to conclude that there are multiple joint measurements for $\cA$.
\end{proof}

Trivial measurements can be argued to be the simplest class of measurements, and therefore it is natural to study its properties first.
Here we have a second motivation to do so: recalling Eq. (\ref{pidgey}), we see that depolarised versions of a measurement can be interpreted as combinations of the original measurement with a trivial one related to it.

\section{Boundary and uniqueness}\label{jotoleague}

We consider now the boundary of the sets $\cP^{m}$ and $\cJ_m$ of general and jointly measurable tuples of $m$ POVMs.
The \textit{boundary} $\partial(\cP^{m})$ is the set of tuples of POVMs $\cA\in\cP^{m}$ for which there exists a perturbation $\cD^\cA$ such that $\cA + \cD^\cA \in \cP^{m}$, and for all $\epsilon>0$, $\cA - \epsilon\cD^\cA \notin \cP^{m}$.
In this case, we say that the perturbation $\cD^\cA$ \textit{witnesses} that $\cA$ lies in the boundary (see Fig. \ref{spearow}).
We define the boundary of $\cJ_m$ analogously, and denote it by $\partial(\cJ_m)$.
Both concepts are natural generalisations of the boundary of the set of measurements, $\partial(\cP)$.
In Ref. \cite{dariano05} it is shown that a measurement lies in $\partial(\cP)$ if and only if it has an effect with a non-trivial kernel, \ie, which is not full-rank. 
Our goal now is to find similar characterisations for $\partial(\cP^m)$ and $\partial(\cJ_m)$.

We recur to the concept of boundary to capture the idea of having flat parts in the "shape" of a convex set, corresponding to non-extremal points.
It is clear from the definition that extremal tuples, either in $\cP^m$ or $\cJ_m$, are a particular case of tuples in the boundary of the given set.
We can even use any perturbation to illustrate this fact, since for any given perturbation $\bD$ for the extremal tuple $\cA$ it follows that $\epsilon\bD$ is also a valid perturbation, for any $\epsilon\in[0,1]$.
Thus if $-\epsilon\bD$ also preserves the POVMness/joint measurability of $\cA$, then it is a symmetric perturbation with this property, contradicting the extremality of $\cA$.

\begin{example}\label{weezing}
Consider the non-extremal tuple $\cA=[(\II,0),(\II/2,\II/2)]$ of Example \ref{lickitung}.
Notice that $\cD^\cA=[(-\II/2,\II/2),(0,0)]$ is a POVM-preserving perturbation for $\cA$, but for any given $\epsilon>0$ the first element of $\cA-\epsilon\cD^\cA$ is $((1+\epsilon/2)\II, -\epsilon\II/2)$, which is not a valid POVM due to its second effect.
Therefore, $\cA\in \partial(\cP(2,2)^{\times 2})$.
Since $\cD$ also preserves the joint measurability of $\cA$, we actually see that $\cA\in\partial(\cJ_2(2,2))\cap \partial(\cP(2,2)^{\times 2})$.
\end{example}

In Example \ref{weezing} we took advantage of the fact that the first POVM $(\II,0)$ of $\cA$ is an extremal measurement.
It is simple to generalise the reasoning above to show that if any POVM of $\cA$ lies in the boundary $\partial(\cP)$, then $\cA$ belongs to the boundary $\partial(\cP^m)$.
This is another particular case of a characterisation that holds for arbitrary Cartesian products.

\begin{theorem}\label{kadabra}
Let $\cA=[\bA^{(1)},\ldots,\bA^{(m)}]\in\cP^m$ be a tuple of measurements.
Then $\cA\in\partial(\cP^m)$ if and only if $\bA^{(j)}\in\partial(\cP)$ for some $j\in\{1,\ldots,m\}$.
\end{theorem}

\begin{proof}
Suppose initially that $\cA\in\partial(\cP^m)$ and let $\cD=[\bD^{(1)},\ldots,\bD^{(m)}]$ be a perturbation such that $\cA+\cD\in\cP^m, \cA-\frac 1 n \cD\notin\cP^m$ for any $n\in\mathbb{N}$.
This means that for every $n\in\mathbb{N}$ there exists a measurement $\bA^{(j)}$ in $\cA$ such that $\bA^{(j)}+\bD^{(j)}\in\cP, \bA^{(j)}-\frac 1 n \bD^{(j)}\notin\cP$.
Since there is a finite number of POVMs in $\cA$, at least one of them satisfies the prior condition for infinitely many values of $n$, implying that for such $j$ we have $\bA^{(j)}-\epsilon\bD^{(j)}\notin\cP$, for any $\epsilon>0$.
In other words, $\bA^{(j)}\in\partial(\cP)$.

On the other hand, suppose without loss of generality that $\bA^{(1)}\in\partial(\cP)$, with a perturbation $\bD$ witnessing that.
Then $\cD=[\bD,0,\ldots,0]$ witnesses that $\cA$ lies in the boundary $\partial(\cP^m)$.
\end{proof}

We now proceed to investigate the boundary $\partial (\cJ_m)$ of the compatible set.
We are able to connect the boundary property of the compatible tuple to such property of the joint measurement.


\begin{theorem}\label{beedrill}
Let $\cA\in\cJ_m(d,n)$ be a compatible tuple.
Then $\cA$ lies in the boundary of the compatible set if and only if every joint measurement for $\cA$ lies in the boundary $\partial(\cP(d,n^m))$ of the set of $n^m$-outcome measurements.
\end{theorem}

\begin{proof}
Here we denote $\cJ_m(d,n)$ by $\cJ_m$ as usual, but $\cP$ will now refer to $\cP(d,n^m)$, the relevant set of measurements in this situation.
Suppose $\cA\in\partial(\cJ_m)$ and let $\bM$ be a joint measurement for $\cA$.
Then there exists a perturbation $\cD$ such that $\cA+\cD\in\cJ_m, \cA-\epsilon\cD\notin\cJ_m$ for all $\epsilon >0$.

We will now construct a perturbation $\bD$ witnessing that $\bM$ lies in the boundary $\partial(\cP)$ based on $\cD$.
Consider the joint measurement $\bM'$ for the tuple $\cA+\cD$ and define $\bS:=\bM'-\bM$.
We see that $\bS$ preserves the marginals of $\bD$, since the elements of $\cD$ are recovered by its marginals,
\beq
\sum_{r\neq j}\sum_{a_r}S_{a_1\ldots (a_j=i)\ldots a_m} = D^{(j)}_i,
\eeq
for $j=1,\ldots, m,i=1,\ldots, n,$ although we have no preservation of the positive semidefinitiveness condition (nor could have, since possibly none of the operators $D^{(j)}_i$ are positive semidefinite).
Notice now that $\bM + \bS = \bM'$ is a valid POVM, which implies that $\bS$ is a POVM-preserving perturbation for $\bM$.
It is straightforward to check that for any $\epsilon>0$ the marginals of $\bM-\epsilon\bS$ yield the elements of $\cA-\epsilon\cD$, but since we know that $\cA-\epsilon\cD\notin\cJ_m$ we conclude that $M_{a_1\ldots a_m}-\epsilon S_{a_1\ldots a_m}\ngeq 0$ for some $a_1,\ldots,a_m$.
Hence $\bM-\epsilon\bS$ is not a valid measurement for any $\epsilon>0$, and $\bM\in\partial(\cP)$.

Assume now that $\cA\notin\partial(\cJ_m)$.
We will show that there exists a joint measurement $\bM$ for $\cA$ such that every effect of it is full-rank, which implies that it does not belong to boundary of $\cP$ \cite{dariano05}.

If $\cA$ is jointly measurable but does not belong to the boundary of $\cJ_m$, then for every JM-preserving perturbation $\cD$ there exists an $\epsilon >0$ such that $\cA':=\cA-\epsilon\cD$ is still a jointly measurable tuple.
Taking $\cD$ to be formed by depolarising perturbations (see Eq. (\ref{magikarp})) of the POVMs in $\bA$, we see that $\cA' = \Phi_{1+\epsilon}(\cA)$, and hence $\cA = \Phi_{\frac{1}{1+\epsilon}}(\cA')$.
Similarly to Eq. (\ref{pidgey}), we can write
\beq\label{alakazam}
\cA = \frac{1}{1+\epsilon}\cA' + \frac{\epsilon}{1+\epsilon}\cA^\tr,
\eeq
where $\cA^\tr=[\bA^{(1),\tr},\ldots,\bA^{(m),\tr}]$ is a tuple of trivial POVMs given by effects $A^{(j), \tr}_i = \tr(A^{(j)}_i)\II/d$.
Since $\cA\in\cJ_m,\cA\notin\partial(\cJ_m)$, and $\cJ_m\subset\cP^m$, we see that $\cA\notin\partial(\cP^m)$.
By Proposition \ref{kadabra}, all the effects of all POVMs in $\cA$ are full-rank (and therefore non-null). 
Hence the same holds for $\cA^\tr$, and the standard joint measurement for trivial POVMs $\bM^\tr$ defined in Eq. (\ref{onix}) have only non-null (and therefore full-rank) effects.

Eq. (\ref{alakazam}) implies that the combination
\beq
\bM = \frac{1}{1+\epsilon}\bM' + \frac{\epsilon}{1+\epsilon}\bM^{\tr}
\eeq
of the joint measurement $\bM'$ for $\cA'$ and a joint measurement $\bM^\tr$ for $\cA^\tr$ is a joint measurement for $\cA$.
We conclude the proof by seeing that each effect of $\bM$ is a sum of a full-rank and a positive semidefinite operators, which therefore is full-rank.
Thus, $\bM\notin\partial(\cP(d,n^m))$.
\end{proof}

Theorem \ref{beedrill} says that if $\cA\in\partial(\cJ_m(d,n))$ then the set of joint measurements for $\cA$ is contained in the boundary $\partial(\cP(d,n^m))$ (Figure \ref{fearow}).
In particular, since this set is convex, we conclude that it contains a unique element or all its elements belong to the same facet of the boundary; otherwise, some joint measurement for $\cA$ would lie in the interior of $\cP(d,n^m)$.

\begin{figure}[h!]
\begin{center}
\includegraphics[width=1\linewidth]{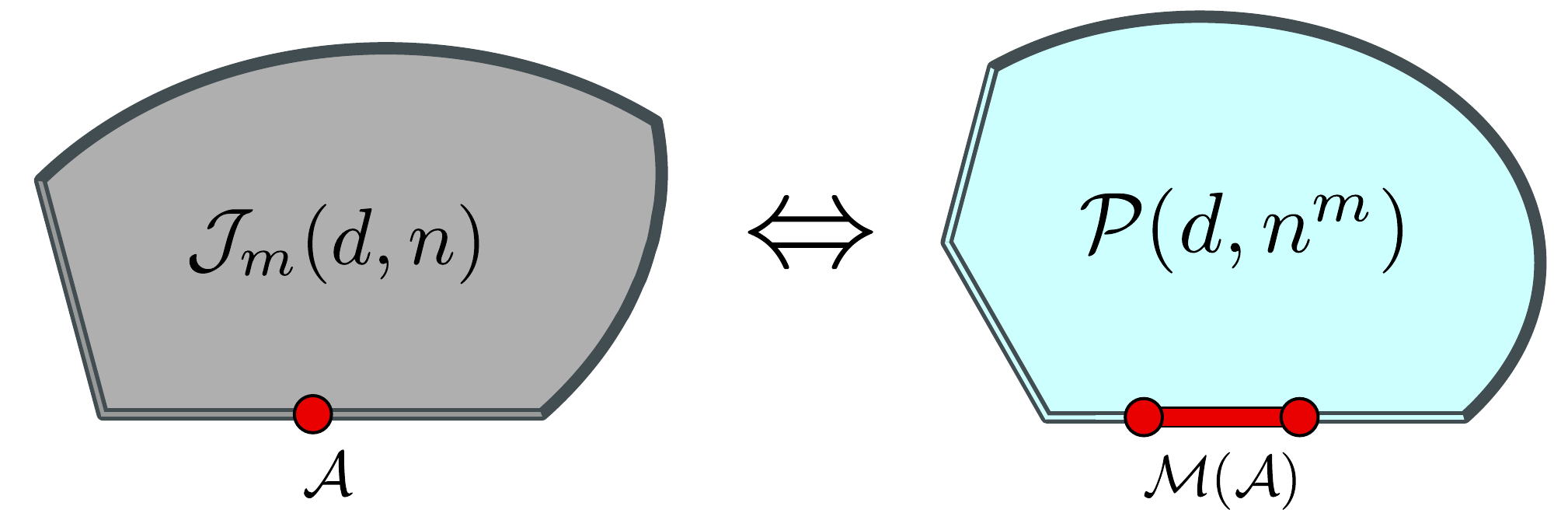}
	\caption{\small{\label{fearow}A pictorial description of Theorem \ref{beedrill}: a tuple of measurements $\cA$ lies in the boundary $\partial(\cJ_m(d,n))$ if and only if the (possibly unit) set $\mathcal{M}(\cA)$ of joint measurements for $\cA$, represented in red, is contained in the boundary $\partial(\cP(d,n^m))$.
	}}
\end{center}
\end{figure}

Our next result establishes that the boundary property is a necessary condition for tuples to have a unique joint measurement.

\begin{theorem}\label{dragonite}
Let $\cA\in\cJ_m(d,n)$ be a compatible tuple of POVMs admitting a unique joint measurement.
Then $\cA$ lies in the boundary of the compatible set $\cJ_m(d,n)$.
\end{theorem}

\begin{proof}
Once more the proof is by contraposition.
Lets assume $\cA\notin\partial(\cJ_m)$.
According to Proposition \ref{kadabra}, this means that $\cA$ does not contain any extremal measurement, and in particular, any deterministic POVM.

Analogously to the previous proof, we can use the fact that $\cA$ is not in the boundary to write it as
\beq
\cA = \frac{1}{1+\epsilon}\cA' + \frac{\epsilon}{1+\epsilon}\cA^\tr,
\eeq
where $\cA^\tr=[\bA^{(1),\tr},\bA^{(m),\tr}]$ is a tuple of trivial POVMs and $\cA'=\Phi_{1+\epsilon}(\cA)$ is another jointly measurable tuple of measurements.
Since there is no deterministic measurement in $\cA$, the same holds for $\cA^\tr$.
Thus $\cA^\tr$ is a tuple of trivial, non-deterministic POVMs, and according to Proposition \ref{electabuzz} there are infinitely many joint measurements $\bM^\tr$.

Thus we have that each combination
\beq
\bM = \frac{1}{1+\epsilon}\bM' + \frac{\epsilon}{1+\epsilon}\bM^{\tr}
\eeq
of a joint measurement $\bM'$ for $\cA'$ and a joint measurement $\bM^\tr$ for $\cA^\tr$ is a joint measurement for $\cA$.
Since each measurement $\bM^\tr$ yields a different $\bM$, we conclude that there is an infinite number of joint measurements for $\cA$.
\end{proof}

Theorem \ref{dragonite} takes advantage of the fact that a tuple in the interior of the compatible set is a noisier version of some other compatible tuple, and this noise allows for a plurality of joint measurements.
The next natural question is whether this condition is also sufficient, what we answer in the negative with the next example.

\begin{example}\label{charizard}
Consider the tuple $\cA^{(xyz)}$ formed by the three dichotomic measurements $\bA^{(x)},\bA^{(y)},\bA^{(z)}$ given by
\begin{equation}
A^{(w)}_{a} = \frac{\II+a\sigma_w}2 ,
\end{equation}
with $a=\pm1, w=x,y,z$, associated to the Pauli observables.
The tuple $\Phi_t(\cA^{(xyz)})$ is jointly measurable for any $t<t^*=1/\sqrt 3$ \cite{heinosaari10}.
At visibility $t^*$, an 8-outcome joint measurement $\bM$ is given by the effects 
\beq
M_{abc} = \frac{1}{8}\left(\II + \frac{a\sigma_x+b\sigma_y+c\sigma_z}{\sqrt 3}\right),
\eeq
$a,b,c=\pm1$.
Since $t^*$ is the critical visibility for $\cA^{(xyz)}$, the tuple $\Phi_{t^*}(\cA^{(xyz)})$ lies in $\partial(\cJ_3(2,2))$, the boundary of the corresponding compatible set.
Indeed, if that was not the case then we would be able use the depolarising perturbation $\cD^\tr$ to find a less-depolarised version of $\cA^{(xyz)}$ still jointly measurable.

However, as already pointed in Ref. \cite{heinosaari10}, this tuple admits more than one joint measurement.
The 8-outcome POVMs $\bM^{(+)},\bM^{(-)}$ given by
\begin{eqnarray}
M^\pm_{abc} = \frac 1 4 \left(\II\pm\frac{a\sigma_x + b\sigma_y + c\sigma_z}{\sqrt 3}\right)
\end{eqnarray}
if $\pm(a,b,c)\in \{(1,1,1),\ (-1,-1,1),\ (-1,1,-1),$ $\ (1,-1,-1)\}$ and $M^\pm_{abc} =0$ otherwise are also joint measurements for $\Phi_{t^*}(\cA^{(xyz)})$, obtained by adding the marginal-preserving perturbation $\bD^\pm$ given by $D^\pm_{abc} = \pm abcM_{abc}$.
(In fact, $\bM^{(+)}, \bM^{(-)}$ are symmetric informationally complete (SIC) tetrahedral measurements embedded in $\cP(2,8)$.)
These POVMs are extremal in the set of joint measurements for $\cA^{(xyz)}$.
We apply a straightforward modification to SDP (\ref{charmander}) to check that there is no marginal-preserving perturbation for $\bM$ orthogonal to $\bD$, hence concluding that the set of joint measurements for $\cA^{(xyz)}$ is completely described by convex combinations of $\bM^{(+)}, \bM^{(-)}$.

Therefore, $\Phi_{t^*}(\cA^{(xyz)})$ has many joint measurements, even though it lies in the boundary of the compatible set.
On the other hand, Theorem \ref{gyarados} tells us that $\cA^{(xyz)}$ is not an extremal tuple in $\cJ_3(2,2)$.
Indeed, we find the decomposition 
\begin{subequations}
\beq
\cA^{(xyz)} = \frac{1} 2 \left(\cB+\cC\right),
\eeq
where $\cB=[\bB^{(1)},\bB^{(2)},\bB^{(3)}]$ is given by
\begin{align}\nonumber
\bB^{(1)} &= \Phi_{\sqrt\frac{2}{3}}\left(\frac{\II+(\sigma_x+\sigma_z)/\sqrt 2}{2}, \frac{\II-(\sigma_x+\sigma_z)/\sqrt 2}{2}\right)\\
\bB^{(2)} &= \Phi_{\frac{1}{\sqrt 3}}(\bA^{(y)})\\ \nonumber
{\bB}^{(3)} &= {\bB}^{(1)}
\end{align}
and $\cC=[\bC^{(1)},\bC^{(2)},\bC^{(3)}]$ is given by
\begin{align}\nonumber
\bC^{(1)} &= \Phi_{\sqrt\frac{2}{3}}\left(\frac{\II+(\sigma_x-\sigma_z)/\sqrt 2}{2}, \frac{\II-(\sigma_x-\sigma_z)/\sqrt 2}{2}\right)\\
\bC^{(2)} &= \Phi_{\frac{1}{\sqrt 3}}(\bA^{(y)})\\ \nonumber
\bC^{(3)} &= \Phi_{\sqrt\frac{2}{3}}\left(\frac{\II-(\sigma_x-\sigma_z)/\sqrt{2}}{2}, \frac{\II+(\sigma_x-\sigma_z)/\sqrt{2}}{2}\right).
\end{align}
\end{subequations}
It is simple to check that $\cB$ and $\cC$ are jointly measurable and extremal in $\cJ_3(2,2)$.
\end{example}

\begin{table*}[t!]
\scriptsize{
\begin{tabular}{|ccccccccccccc|}
\hline
Measurement &&&& Tuple of measurements &&&& Joint measurement &&& Result&\\
\hline
all extremal in $\cP(d,n)$ && $\iff$ && extremal in $\cP(d,n)^{m}$ &&&&&&& Theorem \ref{kubone}&\\

at least one is boundary in $\cP(d,n)$ &&  $\iff$ && boundary in $\cP(d,n)^{m}$ &&&&&&& Theorem \ref{kadabra}&\\

\textcolor{blue}{at least one is extremal in $\cP(d,n)$} && \textcolor{blue}{$\wedge$} && \textcolor{blue}{jointly measurable} && \textcolor{blue}{$\implies$} && \textcolor{blue}{unique}&&& \textcolor{blue}{Corollary \ref{exeggcutor}} &\\

&&&&extremal in $\cJ_m(d,n)$ && $\iff$ && unique, extremal in $\cP(d,n^m)$&&& Corollary \ref{exeggcute}&\\

&&&& extremal in $\cJ_m(d,n)$ && $\centernot \impliedby$ && unique&&& Example \ref{lickitung}&\\

&&&& boundary in $\cJ_m(d,n)$ &&$\iff$&& boundary in $\cP(d,n^m)$&&& Theorem \ref{beedrill}&\\

&&&& boundary in $\cJ_m(d,n)$ &&$\impliedby$&& unique&&& Theorem \ref{dragonite}&\\
&&&& boundary in $\cJ_m(d,n)$ &&$\centernot\implies$&& unique&&& \textcolor{blue}{Example} \ref{charizard}&\\
\hline
\end{tabular}}
\caption{Summary of relations between properties of measurements, tuples of measurements, and joint measurements (in the case where the tuples are jointly measurable). In blue, the previously known results in the literature of quantum measurements~\cite{haapasalo14, heinosaari10}.}
\end{table*}

The above decomposition emphasizes the relevance of tuples in which some measurements are more depolarised than others, such as $\cB$ and $\cC$.

\section{Discussion and conclusion}

We introduced the notion of property-preserving perturbations for tuples of measurements and for joint measurements, and derived a criterion for deciding on the uniqueness of a joint measurement based on it.
By interpreting the depolarising map as the action of a perturbation, we proved that extremality in the set $\cJ_m$ of compatible tuples is a sufficient condition for joint measurement uniqueness, while belonging to the boundary of $\cJ_m$ is a necessary one.
We also provided counter-examples showing that none of these conditions are both necessary and sufficient.

As a by-product, we were led to characterise the extremal and boundary points of the sets of general tuples of measurements $\cP^m$ and of compatible tuples $\cJ_m$, extending well-known results to these so-far unexplored grounds.
With our machinery we also recover previous results in the context of operator algebras, regarding the extremality of the measurements in the tuple~\cite{haapasalo14,haapasalo15}.
(See Table I for a summary of our results.)
It would be interesting to see if the other results we derived here can also be obtained by the operator algebras approach.

In the case where the joint measurement is not unique, a natural question would be to find criteria for deciding on the optimal one.
Each such criterion would provide an unambiguous way of relating each compatible tuple to a unique joint measurement. 

Another raised problem is to relate our questions here to other notions of compatibility, such as coexistence~\cite{lahti03}, as well as to the topic of Einstein-Podolsky-Rosen steering~\cite{wiseman07}, which is closely connected to joint measurability~\cite{quintino14, uola14}.
Using this connection, each joint measurement for a tuple of POVMs $[\bB]$ corresponds to a local-hidden-state (LHS) model for the assemblage $\{\sigma_{i|\bB}=\tr_A(\Psi B^{(i)}\otimes \II)\}$, where $\Psi$ is a fixed full-Schmidt-rank bipartite state.
Then each of our results translates to a relation between such extremal/boundary assemblages and their (possibly unique) LHS models.
A natural next step would be to consider unsteerable states, that is, states for which every assemblage generated by them can be described in terms of an LHS model.
In this scenario, one can ask what is the relation between the boundary/extremal properties in the set of unsteerable states and the uniqueness of the LHS model (for all measurements).
Since not even for the simplest case of Werner states the boundary of such set is completely understood, the techniques presented here might prove to be useful in this context. 

\section*{ACKNOWLEDGEMENTS}

The authors are grateful to Marco T\'ulio Quintino for introducing them to the uniqueness problem, to Jessica Bavaresco and Teiko Heinosaari for fruitful discussions and several improvements for earlier versions of this manuscript, and to Alessandro Toigo and Claudio Carmeli for pointing a flaw in a previous proof of Theorem 3, that restricted it to pairs of POVMs, and proofreading the new proof.
This work was supported by Spanish MINECO (QIBEQI FIS2016-80773-P and Severo Ochoa SEV-2015-0522), Fundaci\`o Cellex, Generalitat de Catalunya (SGR875 and CERCA Program), ERC CoG QITBOX, the Brazilian National Institute for Science and Technology of Quantum Information, as well as the Brazilian agencies CNPq and CAPES.

\bibliographystyle{mod_unsrt}

\begin{thebibliography}{10}

\bibitem{heinosaari16}
T.~Heinosaari, T.~Miyadera, and M.~Ziman.
\newblock An invitation to quantum incompatibility.
\newblock {\itshape Journal of Physics A: Mathematical and Theoretical},
  {\bfseries 49}, 123001, (2016).
\newblock \href{http://arxiv.org/abs/1511.07548}{arXiv: 1511.07548}.

\bibitem{brunner14}
N.~Brunner, D.~Cavalcanti, S.~Pironio, V.~Scarani, and S.~Wehner.
\newblock \href{http://dx.doi.org/10.1103/RevModPhys.86.419}{{Bell}
  nonlocality}.
\newblock {\itshape Rev. Mod. Phys.}, {\bfseries 86}, 419--478, (2014).
\newblock \href{http://arxiv.org/abs/1303.2849}{arXiv: 1303.2849}.

\bibitem{heisenberg27}
W.~Heisenberg.
\newblock \href{http://dx.doi.org/10.1007/BF01397280}{{\"U}ber den
  anschaulichen inhalt der quantentheoretischen {Kinematik} und {Mechanik}}.
\newblock {\itshape Zeitschrift f{\"u}r Physik}, {\bfseries 43}, 172--198,
  (1927).

\bibitem{bennett92}
C.~H. Bennett.
\newblock \href{http://dx.doi.org/10.1103/PhysRevLett.68.3121}{Quantum
  cryptography using any two nonorthogonal states}.
\newblock {\itshape Phys. Rev. Lett.}, {\bfseries 68}, 3121--3124, (1992).

\bibitem{dariano05}
G.~M. D'Ariano, P.~L.~Presti, and P.~Perinotti.
\newblock Classical randomness in quantum measurements.
\newblock {\itshape Journal of Physics A: Mathematical and General}, {\bfseries
  38}, 5979, (2005).
\newblock \href{http://arxiv.org/abs/quant-ph/0408115v2}{arXiv:
  quant-ph/0408115v2}.

\bibitem{teiko08}
T.~{Heinosaari}, D.~{Reitzner}, and P.~{Stano}.
\newblock \href{http://dx.doi.org/10.1007/s10701-008-9256-7}{Notes on joint
  measurability of quantum observables}.
\newblock {\itshape Foundations of Physics}, {\bfseries 38}, 1133-1147, (2008).
\newblock \href{http://arxiv.org/abs/0811.0783}{arXiv: 0811.0783}.

\bibitem{haapasalo14}
E.~Haapasalo, T.~Heinosaari, and J.-P. Pellonp{\"a}{\"a}.
\newblock \href{http://dx.doi.org/10.1142/S0129055X14500020}{When do pieces
  determine the whole? extreme marginals of a completely positive map}.
\newblock {\itshape Reviews in Mathematical Physics}, {\bfseries 26}, 1450002,
  (2014).
\newblock \href{http://arxiv.org/abs/1209.5933}{arXiv: 1209.5933}.

\bibitem{haapasalo15}
E.~Haapasalo, J.-P. Pellonp{\"a}{\"a}, and R.~Uola.
\newblock \href{http://dx.doi.org/10.1007/s11005-015-0754-1}{Compatibility
  properties of extreme quantum observables}.
\newblock {\itshape Letters in Mathematical Physics}, {\bfseries 105},
  661--673, (2015).
\newblock \href{http://arxiv.org/abs/1404.4172}{arXiv: 1404.4172}.

\bibitem{Note1}
Here we consider tuples of measurements (instead of sets) in order to
  unambiguously associate the POVM on the $j$-th entry of the tuple with the
  marginal obtained by summing the effects of the joint measurement over all
  indices but the $j$-th. Tuples are also the appropriate object to define
  extremality of jointly measurable collections of measurements.

\bibitem{wolf09}
M.~M. Wolf, D.~Perez-Garcia, and C.~Fernandez.
\newblock \href{http://dx.doi.org/10.1103/PhysRevLett.103.230402}{Measurements
  incompatible in quantum theory cannot be measured jointly in any other
  no-signaling theory}.
\newblock {\itshape Phys. Rev. Lett.}, {\bfseries 103}, 230402, (2009).
\newblock \href{http://arxiv.org/abs/0905.2998}{arXiv: 0905.2998}.

\bibitem{cavalcanti17}
D.~Cavalcanti and P.~Skrzypczyk.
\newblock \href{http://dx.doi.org/10.1088/1361-6633/80/2/024001}{Quantum
  steering: a short review with focus on semidefinite programming}.
\newblock {\itshape Rep. Prog. Phys.}, {\bfseries 80}, 024001, (2017).
\newblock \href{http://arxiv.org/abs/1604.00501}{arXiv: 1604.00501}.

\bibitem{heinosaari10}
T.~Heinosaari, M.~A. Jivulescu, D.~Reitzner, and Mario Ziman.
\newblock \href{http://dx.doi.org/10.1103/PhysRevA.82.032328}{Approximating
  incompatible von neumann measurements simultaneously}.
\newblock {\itshape Phys. Rev. A}, {\bfseries 82}, 032328, (2010).

\bibitem{lahti03}
P.~Lahti.
\newblock Coexistence and joint measurability in quantum mechanics.
\newblock {\itshape International Journal of Theoretical Physics}, {\bfseries
  42}, 893--906, (2003).
\newblock \href{http://arxiv.org/abs/quant-ph/0211064}{arXiv:
  quant-ph/0211064}.

\bibitem{wiseman07}
H.~M. Wiseman, S.~J. Jones, and A.~C. Doherty.
\newblock \href{http://dx.doi.org/10.1103/PhysRevLett.98.140402}{Steering,
  entanglement, nonlocality, and the {Einstein}-{Podolsky}-{Rosen} paradox}.
\newblock {\itshape Phys. Rev. Lett.}, {\bfseries 98}, 2, (2007).
\newblock \href{http://arxiv.org/abs/quant-ph/0612147}{arXiv:
  quant-ph/0612147}.

\bibitem{quintino14}
M.~T. Quintino, T.~V\'ertesi, and N.~Brunner.
\newblock \href{http://dx.doi.org/10.1103/PhysRevLett.113.160402}{Joint
  measurability, {Einstein}-{Podolsky}-{Rosen} steering, and {Bell}
  nonlocality}.
\newblock {\itshape Phys. Rev. Lett.}, {\bfseries 113}, 160402, (2014).
\newblock \href{http://arxiv.org/abs/1406.6976}{arXiv: 1406.6976}.

\bibitem{uola14}
R.~Uola, T.~Moroder, and O.~G\"uhne.
\newblock \href{http://dx.doi.org/10.1103/PhysRevLett.113.160403}{Joint
  measurability of generalized measurements implies classicality}.
\newblock {\itshape Phys. Rev. Lett.}, {\bfseries 113}, 160403, (2014).
\newblock \href{http://arxiv.org/abs/1407.2224}{arXiv: 1407.2224}.

\end{thebibliography}

\end{document}